\documentclass[conference,a4paper]{IEEEtran}
\usepackage{bbm,dsfont,mathrsfs,fixmath,amsthm}
\usepackage{amsmath,amssymb,graphicx} 
\usepackage{stmaryrd,cite}
\usepackage{amsmath,amstext,amsfonts,amssymb}
\usepackage{epsfig}
\usepackage{exscale}
\usepackage{enumerate}
\usepackage{mathtools}
\usepackage{multirow}
\usepackage[caption=false,font=footnotesize]{subfig}

\newtheorem{definition}{Definition}
\newtheorem{lemma}{Lemma}
\newtheorem{proposition}{Proposition}
\newtheorem{theorem}{Theorem}

\newfont{\bbb}{msbm10 scaled 700}

\newfont{\bb}{msbm10 scaled 1100}

\newcommand{\EE}{\mbox{\bb E}}

\newcommand{\sv}{{\bf s}}

\newcommand{\xv}{{\bf x}}
\newcommand{\yv}{{\bf y}}
\newcommand{\zv}{{\bf z}}

\newcommand{\Cc}{{\cal C}}

\newcommand{\Ec}{{\cal E}}
\newcommand{\Fc}{{\cal F}}

\newcommand{\Ic}{{\cal I}}

\newcommand{\Pc}{{\cal P}}

\newcommand{\Sc}{{\cal S}}
\newcommand{\Tc}{{\cal T}}

\newcommand{\Wc}{{\cal W}}

\newcommand{\Xc}{{\cal X}}
\newcommand{\Yc}{{\cal Y}}
\newcommand{\Zc}{{\cal Z}}

\newcommand{\var}{{\hbox{var}}}

\newcommand{\eqdef}{\stackrel{\Delta}{=}}

\usepackage{hyperref}
\hypersetup{
    bookmarks=true,         
    unicode=false,          
    pdftoolbar=true,        
    pdfmenubar=true,        
    pdffitwindow=false,     
    pdfstartview={FitH},    
    pdfnewwindow=true,      
    colorlinks=true,       
    linkcolor=red,          
    citecolor=green,        
    filecolor=blue,      
    urlcolor=blue           
}

\begin{document}

\title{Joint State Sensing and Communication: \\
Optimal Tradeoff for a Memoryless Case} 
\author{\IEEEauthorblockN{Mari Kobayashi}
\IEEEauthorblockA{
\thanks{The work of M. Kobayashi was supported by an Alexander von Humboldt Research Fellowship.}
Technical University of Munich \\
Munich, Germany\\
 {\tt mari.kobayashi@tum.de}
}
\and
\IEEEauthorblockN{Giuseppe Caire}
\IEEEauthorblockA{
\thanks{The work of M. Kobayashi was supported by an Alexander von Humboldt Research Fellowship.}
Technical University of Berlin \\
Berlin, Germany\\
 {\tt caire@tu-berlin.de}
}
\and
\IEEEauthorblockN{Gerhard Kramer}
\IEEEauthorblockA{
\thanks{The work of M. Kobayashi was supported by an Alexander von Humboldt Research Fellowship.}
Technical University of Munich \\
Munich, Germany\\
 {\tt gerhard.kramer@tum.de}
}
}

\maketitle
\begin{abstract}
A communication setup is considered where a transmitter wishes to simultaneously sense its channel state and convey a message to a receiver. 
The state is estimated at the transmitter by means of generalized feedback, i.e. a strictly causal channel output that is observed at the transmitter. 
The scenario is motivated by a joint radar and communication system where the radar and data applications share the same frequency band. For the case of a memoryless channel with i.i.d. state sequences, we characterize the capacity-distortion tradeoff, defined as the best achievable rate below which a message can be conveyed reliably while satisfying some distortion constraint on state sensing. An iterative algorithm is proposed to optimize the input probability distribution. Examples demonstrate the benefits of joint sensing and communication as compared to a separation-based approach.  
\end{abstract}

\section{Introduction}
A key enabler of autonomous mobile networks is the ability to continuously sense and react to a dynamically changing environment (hereafter called ``state") while letting nodes exchange information with each other. Many existing systems consider an approach based on separation such that resources are divided into either state sensing or data communications. Unfortunately, such a separation-based approach has limitations; i) it performs poorly in high mobility scenarios and for a large state dimension; ii) the data rate degrades by dedicating more resources to state sensing, as no data symbols are sent during the sensing phase. 

These limitations suggest that state sensing and communication should be designed jointly by sharing the same bandwidth.  A number of recent works have studied a joint approach, especially in the context of radar and communication systems operating in-band (see e.g. \cite{sturm2011waveform,bliss2014cooperative,bica2015opportunistic,huang2015radar} and references therein). These works can be roughly classified into interference avoidance and common waveform design \cite{sturm2011waveform,huang2015radar,bica2015opportunistic}. 
Although the latter class considers a joint design, these works mainly apply communication-oriented waveforms such as OFDM to the radar estimation or vice versa. 
Although these works provide waveform design and analysis/simulation of specific scenarios, they do not provide a fundamental framework to study the optimal tradeoff between radar sensing and communication, irrespective of the assumptions on interference avoidance or joint waveforms. We provide a first answer to the optimal tradeoff between state sensing and communication, although restricting to the simplest memoryless case. 

We study the fundamental limits of joint sensing and communication for a point-to-point channel where the transmitter estimates the channel state via a strictly causal channel output, while the receiver has perfect state knowledge. To characterize the tension between sensing quality and communication rate, the capacity-distortion tradeoff is considered as a performance metric. This metric was introduced and studied in~\cite{choudhuri2013causal} and references therein. We remark that the work~\cite{choudhuri2013causal} and the current work differ in their assumptions and concepts. Namely, in \cite{choudhuri2013causal} the transmitter conveys the state to the receiver. In the current work, the transmitter is ignorant of the state and wishes to estimate it using the generalized feedback. The main contributions of the paper are outlined below. 
\begin{enumerate}
\item We characterize the capacity-distortion tradeoff for the memoryless channel with an i.i.d. state sequence;
\item we formulate the capacity-distortion tradeoff maximization as a convex optimization with respect to the input distribution and propose an iterative algorithm; 
\item we provide examples to demonstrate the benefits of a joint design as compared to the separation-based one.  
\end{enumerate}
This paper is organized as follows. We describe the model for joint state sensing and communication in Section \ref{section:model}. Section \ref{section:tradeoff} characterizes the capacity-distortion tradeoff and Section \ref{section:optimization} provides numerical methods to calculate the capacity-distortion tradeoff. We conclude the paper with numerical examples in Section \ref{section:example}.
\section{System Model}\label{section:model}
Consider the communication setup depicted in Fig.~\ref{fig:Model}. 
The channel input, output, feedback output, and state random variables are $X_i$, $Y_i$, $Z_i$, and $S_i$ that take values in the sets $\Xc$, $\Yc$, $\Zc$, and $\Sc$, respectively. The relation between these random variables is characterized by a memoryless channel with i.i.d. states. The joint probability distribution of our model is given by 
\begin{align}
 & P_{WX^nS^nY^nZ^n}(w, \xv, \sv, \yv, \zv) \nonumber\\
& =
P(w) \prod_{i=1}^n P_{S}(s_i)   \prod_{i=1}^n  P(x_i| w z^{i-1})  P_{YZ|XS} (y_iz_i|x_is_i).  
\end{align}
The notation emphasizes that $P_S(\cdot)$ and $P_{YZ|XS}(\cdot)$ are time-invariant.
A $(2^{nR}, n)$ code for the channel consists of
\begin{enumerate}
\item a message set $\Wc = [1:2^{nR}]$;
\item an encoder that sends a symbol $x_i=\phi_i(w, z^{i-1})$ for each message $w\in \Wc$ and each delayed feedback output $z^{i-1}\in \Zc^{i-1}$; 
\item a decoder that assigns a message estimate $\hat{w}=g(y^n, s^n)\in \Wc$;
\item a state estimator that assigns an estimation sequence $\hat{s}^n \in \hat{\Sc}^n$ to each feedback output sequence
$z^n\in \Zc^n$ and the channel input sequence $x^n\in \Xc^{n}$. The set $\hat{\Sc}$ denotes the reconstruction alphabet. 
\end{enumerate}
The state estimate is measured by the expected distortion 
\[
\EE[d(S^n,\hat{S}^n) ]= \frac{1}{n} \sum_{i=1}^n \EE[d(S_i, \hat{S}_i)]
\]
where $d: \Sc\times \hat{\Sc} \mapsto [0,\infty)$ is a distortion function.
A rate distortion pair $(R, D)$ is said to be achievable if there exist $(2^{nR}, n)$ codes with $\lim_{n\rightarrow \infty}P(\hat{W}\neq W) =0$ and $\limsup_{n\rightarrow \infty} \EE[d( S^n,\hat{S}^n)] \leq D$. The capacity-distortion tradeoff $C(D)$ is defined as the supremum of $R$ such that $(R, D)$ is achievable. 

From the well-known result on a memoryless channel with i.i.d. random states where the state is available only at the decoder \cite[Sec. 7.4]{el2011network}, the capacity for the case of unconstrained distortion is 
\begin{align}
C(D=\infty) = \max_{P_X} I(X;Y,S) = \max_{P_X} I(X;Y|S), 
\end{align}
where the maximum is over the input distribution $P_X$. This capacity is achieved by ignoring the feedback. 

\begin{figure}[t]
	\begin{center}	
		\includegraphics[width=0.5\textwidth]{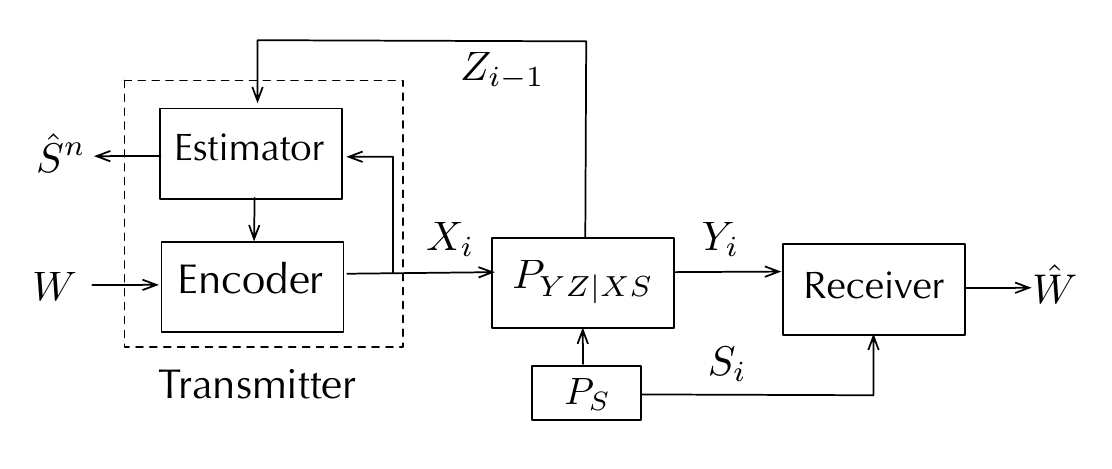}
		\caption{State-dependent channel with generalized feedback}
	\label{fig:Model}
	\vspace{-2em}
	\end{center}
\end{figure}

\section{Capacity-distortion tradeoff}\label{section:tradeoff}
This section characterizes the capacity-distortion tradeoff $C(D)$. We provide some useful lemmas and then the converse and achievability proofs. 
\begin{theorem} \label{th:tradeoff} The capacity-distortion tradeoff of the state-dependent memoryless channel with the i.i.d. states is given by 
\begin{align}
C(D)& = \max I(X;Y|S) 
\end{align}
where the maximum is over all $P_X$ satisfying $ \EE[d( S,\hat{S})] \leq D$ and the joint distribution of $SXYZ\hat{S}$ is given by $P_X(x)P_S(s)P_{YZ|XS} (yz|xs) P_{\hat{S}|XZ}(\hat{s}|xz)$.
\end{theorem}

To prove Theorem \ref{th:tradeoff}, we first provide useful properties of $C(D)$ and the state estimator. 
\begin{lemma}\label{lemma:properties} 
$C(D)$ is a nondecreasing concave function of $D$ for $D\geq D_{\min} \eqdef \min \EE [d(S, \hat{S})]$ where the minimum is over all $P_X$ 
and $P_{\hat{S}|XS}$.
\end{lemma}

\begin{lemma}\label{lemma:Shat} 
We can choose without loss of generality a deterministic estimator given by 
\begin{align}
   \hat{s} = \hat{s}(x,z) = {\rm arg}\min_{s'\in \Sc} \sum_{s\in \Sc} P_{S|XZ}(s|x,z)  d(s, s')
   \label{eq:estimator}
\end{align}
for all $x,z$.
\end{lemma}
\begin{proof}
\begin{align}\label{eq:step1}
& \EE[d(S, \hat{S})] = \EE\left[ \EE[d(S, \hat{S}) |X,Z] \right] \nonumber \\
&\overset{\mathrm{(a)}}= \sum_{x,z} P_{XZ}(xz) \sum_{\hat{s}\in \Sc}  P_{\hat{S}|XZ} (\hat{s}|xz)\sum_{s} P_{S|XZ}(s|xz)  d(s, \hat{s})\nonumber \\
&\geq \sum_{x,z} P_{XZ}(xz) \min_{\hat{s}\in \Sc}\sum_{s} P_{S|XZ} (s|xz)  d(s, \hat{s}) \nonumber\\
&\overset{\mathrm{(b)}}= \sum_{x,z} P_{XZ}(xz) \sum_{s} P_{S|XZ} (s|xz) d(s, \hat{s}(x, z))\nonumber\\
&=\EE[d(S, \hat{s}(X,Z))] 
\end{align}
where $(a)$ follows from the Markov chain $\hat{S}-XZ-S$, and $(b)$ follows by 
choosing \eqref{eq:estimator}. 
\end{proof}

\subsection{Converse} 
 From Fano's inequality we have
\begin{align}
nR & \leq I(W; Y^n, S^n) + n\epsilon_n \nonumber\\
&=  I(W; Y^n| S^n)+ n\epsilon_n\nonumber\\
&= \sum_{i=1}^n H(Y_i| Y^{i-1} S^n) - H(Y_i | W, Y^{i-1} S^n)+ n\epsilon_n \nonumber \\
&\overset{\mathrm{(a)}}\leq  \sum_{i=1}^n H(Y_i|  S_i) - H(Y_i | X_i, Y^{i-1}, W, S^n) + n\epsilon_n \nonumber\\
&\overset{\mathrm{(b)}} =  \sum_{i=1}^n H(Y_i| S_i) - H(Y_i | X_i,  S_i)+ n\epsilon_n\nonumber \\
&= \sum_{i=1}^n I(X_i;Y_i| S_i)+ n\epsilon_n  
\end{align}
where $(a)$ follows by removing the conditioning on $\{S_l\}_{l\neq i}, Y^{i-1}$ in the first term and adding the conditioning on $X_i$ in the second term; $(b)$ follows because $(W,  Y^{i-1}, \{S_l\}_{l\neq i}) - (S_i, X_i) - Y_i $ forms a Markov chain.
We also have
\begin{align}
R & \leq \frac{1}{n} \sum_{i=1}^n I(X_i;Y_i| S_i)+ \epsilon_n \nonumber \\
&\overset{\mathrm{(a)}} \leq \frac{1}{n} \sum_{i=1}^n C\left(\EE[d(S_i, \hat{S}_i)]\right) + \epsilon_n \nonumber\\
&\overset{\mathrm{(b)}} \leq C\left(\frac{1}{n} \sum_{i=1}^n \EE[d(S_i, \hat{S}_i)]\right) + \epsilon_n \nonumber\\
&\overset{\mathrm{(c)}} \leq C(D)
\end{align}
where $(a)$ follows from the definition of $C(D)$; $(b)$ follows from the concavity property of Lemma \ref{lemma:properties}; $(c)$ follows from the nondecreasing property of Lemma \ref{lemma:properties}.
\subsection{Achievability}
We prove Theorem \ref{th:tradeoff} when the distortion function $d(\cdot)$ is bounded by $d_{\rm max}= \max_{(s, \hat{s})\in \Sc\times \hat{\Sc}} d(s, \hat{s})<\infty$. The proof can be extended, as usual, to $d(\cdot)$ for which there is a letter $s^*$ such that $\EE[d(S, s^*)] \leq d_{\max}$.
 \paragraph{Codebook generation} 
 Fix $P_X(\cdot)$ and functions $\hat{s}(x,z)$ that achieve $C(D/(1+\epsilon))$, where $D$ is the desired distortion.
 Randomly and independently generate $2^{nR}$ sequences $x^n(w)$ for each $w\in [1:2^{nR}]$.  This defines the codebook $\Cc$ which is revealed to the encoder and the decoder. 
 \paragraph{Encoding} To send a message $w\in \Wc$, the encoder chooses and transmits $x^n(w)$. 
 \paragraph{Decoding} The decoder finds a unique index $\hat{w}$ such that $(y^n, s^n, x^n(\hat{w}))$ is jointly typical, i.e. 
\begin{align}
 (y^n, s^n, x^n(\hat{w})) \in \Tc_{\epsilon}^{(n)} 
 \end{align}
 \paragraph{Estimation} The encoder computes the reconstruction sequence as $
 \hat{s}^n = \hat{s}(x^n(w),z^n)$. 
 \paragraph{Analysis of Expected Distortion}
In order to bound the distortion averaged over a random choice of the codebooks $\Cc$, we define the decoding error event. The decoder makes an error if and only if one or both of the following events occur. 
  \begin{align}
\Ec_1 &= \{  (y^n, s^n, x^n(1)) \notin \Tc_{\epsilon}^{(n)}  \} \\
\Ec_2 &= \{  (y^n, s^n, x^n(\hat{w})) \in \Tc_{\epsilon}^{(n)} \;\; \text{for some $w\neq 1$} \}
\end{align}
 where we assume without loss of generality that $w=1$ is sent. By the union bound, we have
 \begin{align*}
 P_e = P(\Ec_1 \cup \Ec_2 )  \leq P(\Ec_1) +  P( \Ec_2).
 \end{align*}
 The first term goes to zero as $n\rightarrow \infty$ by the law of large numbers. The second term also tends to zero as $n\rightarrow \infty$ if $R < I(X;Y|S)$ by the independence of the codebooks and the packing lemma \cite[Lemma 3.1]{el2011network}. Therefore, $ P_e$ tends to zero as $n\rightarrow \infty$ if $R < I(X;Y|S) $. 
 
 If there is no decoding error, we have 
\begin{align}
(Y^n, S^n, X^n(1)) \in \Tc_{\epsilon}^{(n)}.
\end{align}
The expected distortion averaged over the random codebook, encoding and decoding, is upper bounded as
  \begin{align}
  \lefteqn{\limsup_{n\rightarrow \infty} \EE[d(S^n, \hat{S}^n)]}\nonumber \\
 &\overset{\mathrm{(a)}}  \leq  \limsup_{n\rightarrow \infty} 
 \left( P_e d_{\max} +(1+\epsilon) (1-P_e) \EE[d(S, \hat{S})] \right) \nonumber\\
 &\overset{\mathrm{(b)}} \leq   \limsup_{n\rightarrow \infty}   \left( P_e d_{\max} + (1+\epsilon) (1-P_e) D\right) \nonumber\\
 &\overset{\mathrm{(c)}} = (1+\epsilon)D
 \end{align}
where $(a)$ follows by applying the upper bound on the distortion function to the decoding error event and the typical average lemma \cite[Ch. 2.4]{el2011network} to the successful decoding event; $(b)$ follows from the assumption on $P_X$ and $\hat{s}(x, z)$ that satisfy $\EE[d(S, \hat{S})]\leq D$; $(c)$ follows because $\limsup_{n\rightarrow \infty} P_e\rightarrow 0$ if $R < I(X;Y|S)= C(\frac{D}{1+\epsilon})$, which proves the achievability of the pair $(C(\frac{D}{1+\epsilon}), D)$. Finally, by the continuity of $C(D)$ in $D$, $C(D)$ is achieved as $\epsilon\rightarrow 0$. 

\section{Numerical Method for Optimization}\label{section:optimization}
Suppose the channel input is subject to a cost constraint $B$ in addition to the distortion constraint $D$. That is, we consider a cost function $b(X^n) = \frac{1}{n}\sum_{i=1}^n b(X_i)$ such that $\limsup_{n\rightarrow \infty} \EE[b(X^n)] \leq B$. Following similar steps above, one can show that the capacity-distortion-cost tradeoff is 
\[
C(D, B)= \max I(X;Y|S)
\]
where the maximum is over $P_X$ satisfying both the distortion and cost constraints. 
We formulate the above maximization as a convex optimization with respect to the input distribution and propose a modified Blahut-Arimoto algorithm. 
\subsection{Problem Formulation} 
The optimization problem can be stated as
\begin{subequations} \label{P1} 
\begin{eqnarray}  
\mbox{maximize} & & I(X;Y|S) \label{obj} \\
\mbox{subject to} & & \EE[b(X)] \leq B \label{cost} \\
& & \EE[d(S,\hat{S})] \leq D. \label{distortion}
\end{eqnarray}
\end{subequations}
For the joint distribution 
$P_X P_S P_{YZ|XS} P_{\hat{S}|XZ}$, 
the estimator $\hat{s}(x, z)$ given in Lemma \ref{lemma:Shat} can be computed a priori. 
At this point, the original problem \eqref{P1} can be rewritten explicitly 
in terms of $P_X$ 
\begin{subequations} \label{P2} 
\begin{eqnarray}  
\mbox{maximize} & &   \Ic(P_X, P_{Y|XS} | P_S)\label{obj2} \\
\mbox{subject to} & & \sum_x b(x) P_X(x)  \leq B \label{cost2} \\
& & \sum_x c(x) P_X(x) \leq D \label{distortion2} \\
& & \sum_x P_X(x) = 1  \label{pmf} 
\end{eqnarray}
\end{subequations}
where we define the mutual information functional
\begin{small}
\begin{align} \label{I}
& \Ic(P_X, P_{Y|XS} | P_S) \nonumber\\
 &= \sum_s P_S(s) \sum_x \sum_y P_X(x) P_{Y|XS}(y|xs) \log \frac{P_{Y|XS}(y|xs)}{P_{Y|S}(y|s)}. 
\end{align}
\end{small}

We define an additional cost function
\begin{equation} \label{newcost}
c(x) = \sum_{z\in \Zc} P_{Z|X}(z|x) \sum_{s \in \Sc} P_{S|XZ} (s|xz) d(s, \hat{s}(x, z)).
\end{equation}
Two remarks are in order. First, the problem \eqref{P2} has two cost functions that must simultaneously satisfy their corresponding average constraints. By letting $\Fc(B, D)$ denote a feasible set of $P_X$ satisfying constraints \eqref{cost2} and \eqref{distortion2}, $\Fc(B, D)$ might be empty for some values of $(B, D)$. For simplicity, we assume that $\Fc(B, D)$ is not empty and therefore $\Fc(B, D)$ is
a convex compact set (the constraints are both linear). Second, the solution of \eqref{P2} does not generally satisfy both constraints  \eqref{cost2} and \eqref{distortion2} with equality. Thus, it is reasonable to consider a parametric form of the optimization problem by incorporating one cost function as a penalty term in the objective function and focusing on the other cost function. The new problem is to maximize 
\begin{eqnarray}  
 \Ic(P_X, P_{Y|XS} | P_S) - \mu  \sum_x c(x) P_X(x)  \label{obj3} 
 \end{eqnarray}
subject to the constraints \eqref{cost2} and \eqref{pmf}, where $\mu \geq 0$ is a fixed parameter. 

 We proceed as in the derivation of the standard Blahut-Arimoto algorithm \cite{blahut1972computation} that computes the capacity-cost function. 
In particular, let $Q_{X|YS}$ denote a general conditional pmf of $X$ given $Y,S$ and consider the function
\begin{align*} 
\lefteqn{J(P_X,P_{Y|XS}, Q_{X|YS} | P_S) }\nonumber\\
& =  \sum_s P_S(s) \sum_x \sum_y P_X(x) P_{Y|XS}(y|xs) \log \frac{Q_{X|YS}(x|ys)}{P_X(x)}. 
\end{align*}
We then have the following result. 
\begin{theorem} \label{modifiedBA}
Let $\Pc(B)$ denote the set of pmfs $P_X$ satisfying (\ref{cost2}). The following statements hold: \\
a) Letting $L(B,D)$ denote the optimal value of \eqref{obj3}, we have
\begin{align} 
L(B,D) = \max_{P_X \in \Pc(B)}  \min_{Q_{X|YS}} J(P_X,P_{Y|XS}, Q_{X|YS} | P_S)  \nonumber \\
 - \mu  \sum_x c(x) P_X(x).  \label{BAa}
\end{align} 
b) For fixed $P_X \in  \Pc(B)$,  the function $J - \mu  \sum_x c(x) P_X(x)$ is maximized over $Q_{X|YS}$ by 
\begin{equation} 
Q^\star_{X|YS}(x|ys) = \frac{P_X(x) P_{Y|XS}(y|xs)}{\sum_{x'} P_X(x') P_{Y|XS}(y|x's)}.  \label{BAb}
\end{equation}
c) For fixed $ Q_{X|YS}$, the function $J- \mu  \sum_x c(x) P_X(x)$
is maximized over $P_X \in \Pc(B)$ by 
\begin{equation} \label{BAc}
P_X^\star(x) = \frac{e^{g(x)}}
{\sum_{x'} e^{g(x')}}
\end{equation}
where 
\begin{small}
\begin{equation}
g(x) =  \sum_s \sum_y P_S(s) P_{Y|XS}(y|xs) \log Q_{X|YS}(x|ys) - \lambda b(x) - \mu c(x) 
\end{equation}
\end{small}
and $\lambda \geq 0$ is chosen so that \eqref{cost2} holds with equality. \hfill $\square$
\end{theorem}

The proof follows by a standard alternating optimization technique (see \cite[Chapter 13.8]{cover_book}) and is omitted due to space limitations. 

\subsection{Modified Blahut-Arimoto Algorithm} 
Based on a general result on alternating optimization \cite[Chapter 13.8]{cover_book}, we propose the following algorithm: 
\begin{itemize}
\item Initialization: \\
Fix $\mu \geq 0$, and let $P^{(0)}_X(x) = \frac{1}{|\Xc|}, \forall x\in \Xc$. 
\item For $k = 1,2,3, \ldots$ do:
\begin{enumerate}
\item Let
\begin{small}
\begin{equation} 
Q^{(k)}_{X|YS}(x|ys) = \frac{P^{(k-1)}_X(x) P_{Y|XS}(y|xs)}{\sum_{x'} P^{(k-1)}_X(x') P_{Y|XS}(y|x's)}.  \label{BAstep1}
\end{equation}
\end{small}
\item 
Choose $\lambda^{(0)} > 0$ and, for $\ell = 1,2,\ldots$, repeat:\\
 a) compute primal variables:  
 $p^{(\ell)}(x) = \frac{e^{g^{(\ell)}(x)}}{\sum_{x'} e^{g^{(\ell)}(x')}}$ with
\begin{small}
 \begin{align} \label{BAstep21}
g^{(\ell)}(x) =  \sum_{s,y}  P_S(s) P_{Y|XS}(y|xs) \log Q^{(k)}_{X|YS}(x|ys) \nonumber \\
- \lambda^{(\ell-1)} b(x) - \mu  c(x)
\end{align}
\end{small}
b) update dual variables:
\begin{small}
\begin{eqnarray} 
\lambda^{(\ell)} & = & \left [ \lambda^{(\ell-1)} + \alpha_\ell  \left ( \sum_x b(x) p^{(\ell)} (x) - B \right ) \right ]_+ 
\end{eqnarray}
\end{small}
where $\alpha_\ell $ is the gradient adaptation step. Let $P^{(k)}_X(x) =\lim_{\ell \rightarrow \infty}  p^{(\ell)}(x)$. 
\end{enumerate} 
\end{itemize}
The above algorithm yields for any fixed input cost $B$ a pair of capacity-distortion values $(C_\mu(B) , D_{\mu}(B))$.
By 
letting $P^{\infty}_X$ denote the convergent input distribution produced by the algorithm, we have
\begin{subequations}
\begin{eqnarray}
C_\mu(B) & = & \Ic(P^{(\infty)}_X, P_{Y|XS} | P_S) \label{rate-mu} \\
D_\mu(B) & = & \sum_x c(x) P^{(\infty)}_X(x). \label{distortion-mu}
\end{eqnarray}
\end{subequations}
By varying $\mu$, we obtain a capacity-distortion tradeoff for fixed input cost $B$. 
Note that for $\mu = 0$ we obtain the standard capacity-cost function of the channel (disregarding the distortion), while 
as $\mu \rightarrow \infty$ the problem becomes a distortion minimization for a given input cost $B$.

\section{Examples}\label{section:example}
We provide two examples to illustrate the gain of our proposed joint scheme with respect to the separation-based approach.
\begin{definition}
A separation-based approach refers to a scheme whose resources are divided into either state sensing with generalized feedback 
or data communication without feedback. 
\end{definition}
For simplicity, the generalized feedback is assumed to be perfect $Z=Y$. 
\subsection{Binary Channel with Multiplicative Bernoulli State}
Consider a binary channel $ Y = S X$ 
where the state $S$ is Bernoulli distributed such that $P_S(1)\eqdef q \in [0,1/2]$, and the multiplication is binary such that $y =0$ if either $s$ or $x$ is 0 while
$y=1$ if $x=s=1$. We use the Hamming distortion measure $d(s, \hat{s}) = s \oplus \hat{s}$. We characterize the input distribution $P_X(0)\eqdef p \in [0,1/2] $ that maximizes $C(D)$. 
The two extreme points on the capacity-distortion tradeoff are as follows.
\begin{enumerate}
\item If $p=0$ (by sending always $X=1$), the minimum distortion $D_{\min}=0$ is achieved, but we have $C(D)=0$.
\item If $p=1/2$, then we have $C(D_{\max})=q$ for $D_{\max} = q/2$.
\end{enumerate}
More generally, we have the following result.
\begin{proposition} The capacity-distortion tradeoff of the binary channel with the multiplicative Bernoulli state is given by
\begin{align}
C(p) = q H_2(p), \;\; D(p) = q p
\end{align}
where $H_2(p)$ denotes the binary entropy function. 
\end{proposition}

\begin{proof}
Because $Y$ is deterministic given $S, X$, the capacity can be expressed as a function of $p,q$ as follows:
\begin{align*}
 C(p) &=  H(Y|S)\\
 &= -\sum_s P_S(s) \sum_y P_{Y|S} (y|s) \log P_{Y|S} (y|s)= q H_2(p)
\end{align*}
where the last equality follows by $P_{Y|S}(0|0)=1$ and $P_{Y|S}(1|1)=p$. 
To calculate the distortion, we first determine the estimator $\hat{s} (x, y)$ and the resulting cost function $c(x)$.  From Lemma \ref{lemma:Shat}, we have
\begin{align}
\hat{s}(x, 0)= 0, \;\;\forall x, \;\; \hat{s}(x, 1)=1,\;\;\forall x.
\end{align} 
In fact, since $y=1$ cannot be generated from the input $x=0$, the value of $\hat{s}(0, 1)$ is irrelevant. Using \eqref{newcost} and and the conditional pmf, we have
\begin{align}
P_{S|XY} = \begin{cases}
1-q &  \text{if $(x,y,s)=(0,0,0)$}\\
q & \text{if $(x,y, s) = (0,0,1)$} \\
1 & \text{if $(x,y,s) = (1,0,0)$ or $(1,1,1)$}\\
0 & \text{else}
\end{cases}
\end{align}
and we obtain the cost function $c(1)=0$ and $c(0) =q$. This yields the desired distortion.  
\end{proof}
Fig.~\ref{fig:Example} plots $C(D)$ for the case $q=0.4$. Observe that the joint approach yields a significant gain over the separation-based approach that achieves a time-sharing between $(D, C)=(0, 0)$ and $(q,q)$. Note that the distortion $q$ is achieved by considering a fixed estimator $\hat{s}=0$ independent of feedback. 
\begin{figure}[t!]
	\begin{center}	
		\includegraphics[width=0.52\textwidth]{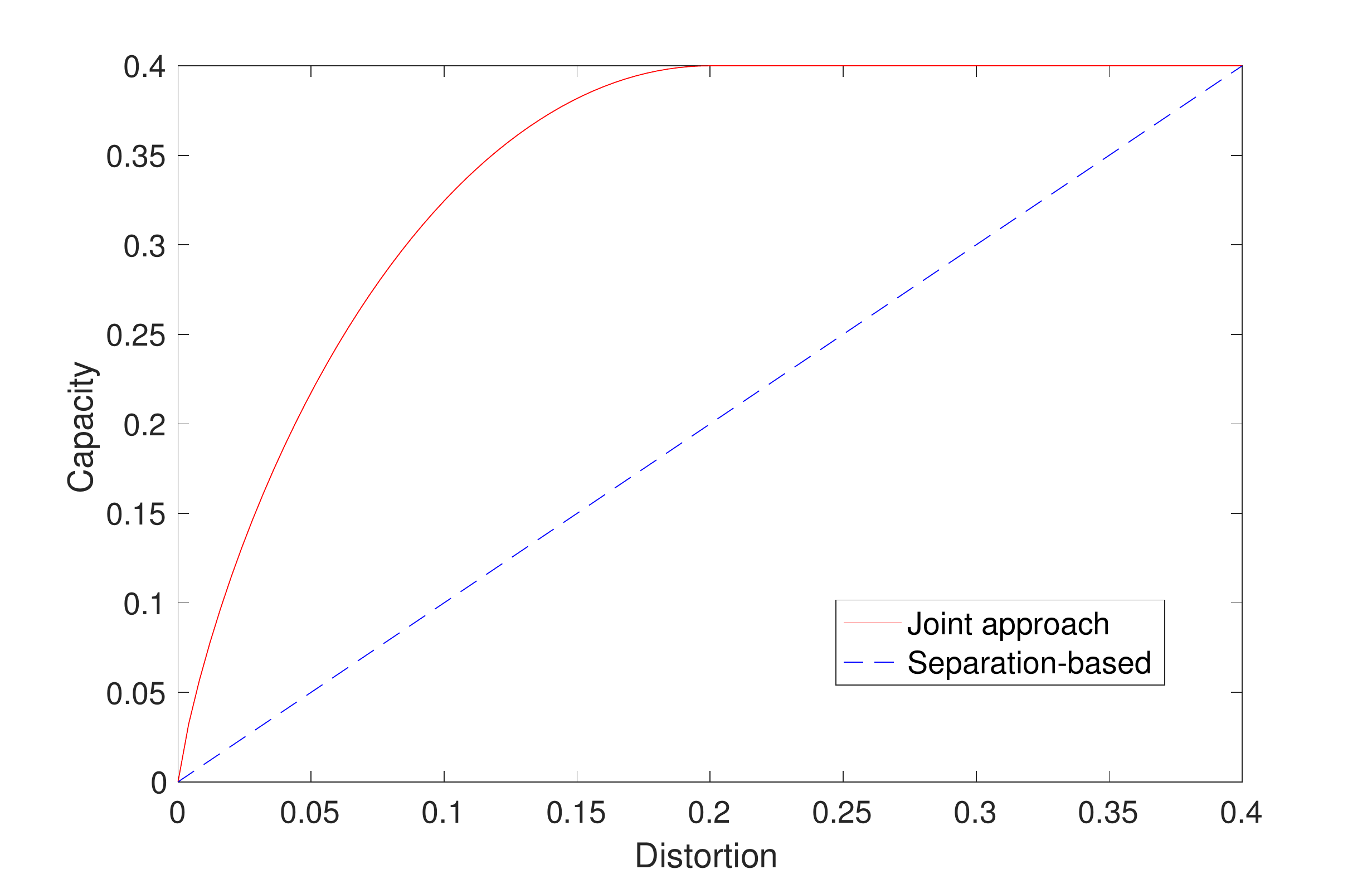}
		\caption{Capacity-distortion tradeoff of binary channel with $q=0.4$.}
	\label{fig:Example}
		\includegraphics[width=0.52\textwidth]{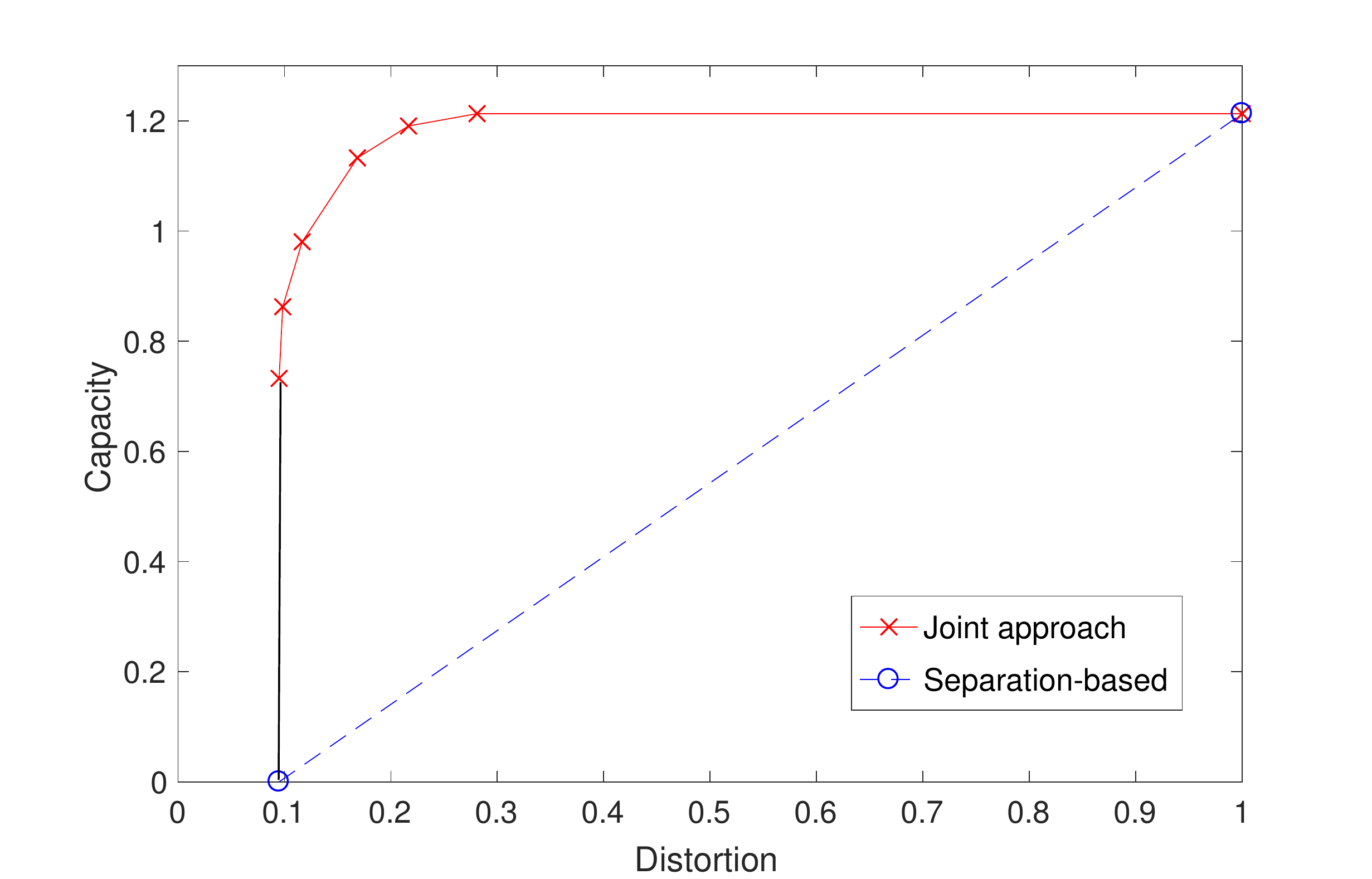} 
		\caption{Capacity-distortion tradeoff of fading AWGN channel $P=10$ dB.}
	\label{fig:Gaussian}
	\end{center}
\end{figure}
\subsection{Real Gaussian Channel with Rayleigh Fading}
Next we consider the real Gaussian channel with Rayleigh fading. The output is given by  
\begin{align}
Y_i = S_i X_i + N_i
\end{align}
where $X_i$ is the channel input satisfying $\frac{1}{n}\sum_i\EE[|X_i|^2] \leq P$, and both $N_i$ and $S_i$ are i.i.d. Gaussian distributed with zero mean and unit variance. 
We let $P=10$ dB. Focusing on the quadratic distortion measure, we consider two extreme cases.
\begin{enumerate}
\item If we relax the distortion constraint, a Gaussian input maximizes the capacity. The unconstrained capacity, denoted by $C_{\max}$, is $\frac{1}{2} \EE[\log (1+|S|^2 P)]=1.213$ [bit/channel use] by averaging over all possible fading states. The corresponding expected distortion is $\EE[\frac{1}{1+|X|^2}]$ where the expectation is with respect to the Gaussian distributed $X$. 
\item The minimum distortion $D_{\min}$ is achieved by $2$-ary pulse amplitude modulation (PAM) and is equal to $\frac{1}{1+P}=0.091$. The corresponding capacity of $2$-PAM is $0.733$ [bit/channel use].
\end{enumerate}
Fig. \ref{fig:Gaussian} shows the capacity-distortion tradeoff calculated by applying the modified Blahut-Arimoto algorithm to the quantized real AWGN channel and $M$-ary PAM. The separation-based approach achieves two corner points, namely $(D_{\min}, 0)$ by dedicating
full resources to state estimation and $(D_{\max}, C_{\max})$ with $D_{\max} \eqdef \var[S]=1$, by ignoring feeback and sending data with Gaussian distribution. 

\bibliographystyle{IEEEbib}

\end{document}